\documentclass[9pt]{article}
\usepackage{algorithmic}
\usepackage{algorithm}
\usepackage{graphicx,url,cite,amsmath,amsthm,amssymb,amsfonts,mathrsfs}

\newcommand{\algo}[1]{\texttt{#1}}

\begin{document}

\title{On Computing Maximal Independent Sets of Hypergraphs in Parallel}

\author{
Ioana O. Bercea \thanks{Department of Computer Science, University of Maryland, USA} 
\and Navin Goyal \thanks{Microsoft Research India}
\and David G. Harris \thanks{Department of Applied Mathematics, University of Maryland, USA} 
\and Aravind Srinivasan \thanks{Department of Computer Science, University of Maryland and UMIACS, USA}
}

\maketitle
\begin{abstract}
Whether or not the problem of finding maximal independent sets (MIS) in hypergraphs is in \emph{(R)NC} is one of the fundamental problems in the theory of parallel computing. Unlike the well-understood case of MIS in graphs, for the hypergraph problem, our knowledge is quite limited despite considerable work.
It is known that the problem is in \emph{RNC} when the edges of the hypergraph have constant size. For general hypergraphs with $n$ vertices and
$m$ edges, the fastest previously known algorithm works in time $O(\sqrt{n})$ with \newline $\text{poly}(m,n)$ processors.
In this paper we give an EREW PRAM algorithm that works in time $n^{o(1)}$ with $\text{poly}(m,n)$ processors on general 
hypergraphs satisfying $m \leq n^{\frac{\log^{(2)}n}{8(\log^{(3)}n)^2}}$ where $\log^{(2)}n = \log\log n$ and $\log^{(3)}n = \log\log\log n$.  
Our algorithm is based on a sampling idea that reduces the dimension of the hypergraph and employs 
the algorithm for constant dimension hypergraphs as a subroutine. 
\end{abstract}

\newpage

\section{Introduction}
Fast parallel algorithms for constructing maximal independent sets (MIS) in graphs are well studied and very efficient algorithms
are now known (see, e.g., \cite{KarpR90} for a brief survey). These algorithms serve as a primitive in numerous applications. The more general problem of fast parallel MIS in hypergraphs is also well studied but is not so well understood. Let us first formally state the 
problem before describing what is known about it.  

A \emph{hypergraph} $H=(V,E)$ is a set of vertices $V$ and a collection of edges $e \in E$ such that $e \subseteq V$. The \emph{dimension} of a hypergraph is the maximum edge size. Let $n$ be the number of vertices, $m$ the number of edges and $d$ the dimension of the hypergraph. A subset of vertices of $H$ is called \emph{independent} in $H$ if it contains no edge. We call an independent set \emph{maximal} if it is not contained in a larger independent set. Karp and Ramachandran~\cite{KarpR90} asked whether the problem of finding an MIS in hypergraphs is in NC. 
While the general problem remains open, progress has been made on some special classes of hypergraphs. We now briefly survey some 
of the previous work; further references can be found in the papers mentioned below.

In a seminal paper, Beame and Luby~\cite{Beame1990} gave an algorithm (called the {\bf \algo{\textsc{BL}} algorithm} henceforth) and 
showed 
that the problem is in RNC for hypergraphs with edges of size at most 3 (\cite{Beame1990} claimed that their algorithm was in RNC for all constant dimension hypergraphs; this however turned out to be erroneous). This algorithm is similar to some of the MIS algorithms for graphs and is based on independently marking vertices and unmarking if all vertices in an edge get marked. 
Kelsen~\cite{Kelsen1992} extended the analysis of the \algo{\textsc{BL}} algorithm to hypergraphs with constant dimension (the dimension can actually be super-constant; we state the precise bound later in the paper where we use this fact). Luczak and Szymanska~\cite{LuczakS97} showed that the problem is in RNC for linear hypergraphs (linear hypergraphs satisfy $|e \cap e'|\leq 1$ for all distinct edges $e, e'$). Beame and Luby~\cite{Beame1990} also gave another appealing algorithm based on random permutations which they conjectured to work in RNC for the general problem. Shachnai and Srinivasan~\cite{ShachnaiS04} made progress towards the 
analysis of this algorithm. For general hypergraphs, Karp, Upfal and Wigderson~\cite{Karp1988} gave an algorithm with running time $O(\sqrt{n})$ and $\mathrm{poly}(m,n)$ processors (their algorithm actually works in a harder model of computation where the hypergraph is accessible only via an oracle, but it can be adapted to run in time $O(\sqrt{n}) \cdot (\log n+\log m)$ with high probability on $mn$ processors).

\paragraph{Our contribution} We give a parallel algorithm that we call the \algo{\textsc{SBL}} (\emph{sampling \algo{\textsc{BL}}}) algorithm. The algorithm works on hypergraphs that do not have too many edges but no other restrictions and works in time $O(n^{o(1)})$. This is the first parallel algorithm that works on general hypergraphs with a relatively weak restriction on the cardinality of the edge set and a running time of $o(\sqrt{n})$. 

More precisely, 
\newtheorem{theorem}{Theorem}
\begin{theorem}
The \algo{\textsc{SBL}} algorithm finds a maximal independent set in hypergraphs with $n$ vertices and $m$ edges and
$m \leq n^{\frac{\log^{(2)}n}{8(\log^{(3)}n)^2}}$. It runs in time $O(n^{2/\log^{(3)}n})$ on EREW PRAM with 
$\mathrm{poly}(m,n)$ processors. 
\end{theorem}
The parameters above have been chosen to keep the computation in the analysis simple and there is some
flexibility in their choice. 

Our algorithm crucially uses \algo{\textsc{BL}} as a subroutine. However, we need to use it on hypergraphs with slightly superconstant
dimensions. Kelsen's original analysis \cite{Kelsen1992} of \algo{\textsc{BL}} is formulated for constant dimension hypergraphs. A slight modification of this analysis, specifically in the potential function used to describe progress being made in each round, allows it to be 
applicable without the assumption that the dimension is $O(1)$. We present this modification. 
We also discuss an additional improvement that can be made to Kelsen's analysis of the  \algo{\textsc{BL}} algorithm, which could be of 
independent interest:  Kelsen 
developed 
concentration inequalities for polynomials in independent random variables. Since then much stronger versions of such inequalities have
become available \cite{KimVu, SchudySviridenko}.  We employ one such inequality and obtain an improved upper bound which we later use in the analysis.
Unfortunately, the above modification does not lead to a significant improvement in the final running time of the algorithm.
Nevertheless, we hope that this identifies the main bottlenecks in Kelsen's approach and will be useful for the future work. 

\paragraph{Organization} The next section is devoted to the \algo{\textsc{SBL}} algorithm. Section 3 delves into Kelsen's analysis of the
\algo{\textsc{BL}} algorithm. We note that there is a large overlap with Kelsen's paper in Section 3 owing to the fact that we are mainly 
talking about modifications to his analysis and this requires us to restate many of his results and proofs to make the paper somewhat
self-contained. 

\section{\algo{\textsc{SBL}} Algorithm}

We now explain the \algo{\textsc{SBL}} algorithm which mainly uses the \algo{\textsc{BL}} algorithm as a subroutine. Denote the input hypergraph by $H = (V, E)$. Intuitively, we can think of  the \algo{\textsc{BL}} algorithm as iteratively coloring the vertices in $V$ red or blue; at the end of the
run of the algorithm the blue vertices will form the final MIS. The idea of our algorithm is to randomly sample a subset $V'$ of vertices by independently marking each vertex in $V$ with probability $p$ (to be carefully chosen). With high probability, the hypergraph $H' = (V', E')$, where $E' = \{e \in E : e \subseteq V'\}$ is the set of edges with all vertices marked, has dimension at most $d$, where $d$ is suitably small (if $H'$ has an edge with size more than $d$ then we declare failure and start over). We then apply the \algo{\textsc{BL}} algorithm to $H'$ to get a red-blue coloring of its vertices, where blue vertices form an MIS in $H'$. This coloring will be the permanent coloring of the vertices of $H'$. Going back to $H$, we remove the edges of $H$ that have a red vertex as these edges cannot be all blue in any completion of the coloring of $H'$. For the remaining edges, we remove their blue vertices and thus get a hypergraph on $V \setminus V'$. We repeat the above process on this updated hypergraph until the number of edges becomes at most $1/p^2$. At this point we can just use the algorithm that takes time linear in the number of vertices or alternatively the Karp--Upfal--Wigderson algorithm~\cite{Karp1988} which we shall call \algo{\textsc{KUW}}.  The vertices in the MIS returned by this last call will again be colored blue, while the rest will be colored red.

\begin{algorithm}
\caption{\algo{\textsc{SBL}}}

\textbf{Input:} A hypergraph $H = (V,E)$ \\
\textbf{Output:} A maximal independent set $I \subseteq V$

~\

\begin{algorithmic}[1] 
\STATE Let $p = 1/n^{\alpha}$ and  $d = \frac{\log^{(2)}n}{4 \log^{(3)}n}$ where $n = |V|$ and $\alpha = 1/\log^{(3)}n$.
\STATE $I \leftarrow \emptyset$
\IF{$ \max_{e \in E} |e| > d$}
	\WHILE {$|V| \geq 1/p^2$}
		\STATE{\textit{Invariant:}  If $I'$ is an IS in $H'$, then $I \cup I'$ will be an IS in $H$. } 
		\STATE Select vertices independently at random with probability $p$ 
		\STATE Let $V'$ be the collection of such selected vertices, $E' = \{e \in E : e \subseteq V'\}$ and $H' = (V',E')$.
		\IF{$ \max_{e' \in E'} |e'| > d$}
			\STATE FAIL 
		\ELSE
			\STATE Run $I' = $\algo{\textsc{BL}}$(H')$ 
			\STATE Update $I=I \cup I'$, $V = V \setminus V'$ 
			\FORALL {$e \in E$}
							\IF {$e \cap (V' \setminus I') \neq \emptyset$}
								\STATE $E \leftarrow E \setminus e$.
							\ENDIF
			\ENDFOR
			\FORALL {$e \in E$}
							\STATE $e \leftarrow e \setminus I'$.
							
			\ENDFOR
		\ENDIF
	 \ENDWHILE
	 \STATE Run $I'= \algo{\textsc{KUW}}(H$).
	 \STATE Update $I=I \cup I'$.
\ELSE
	\STATE Run $I = $\algo{\textsc{BL}}$(H)$.
\ENDIF 
	\STATE Return $I$.
\end{algorithmic}
\end{algorithm}

\subsection{Correctness of \algo{\textsc{SBL}}} We claim that in the final coloring produced by the \algo{\textsc{SBL}} algorithm, the set of blue vertices forms an MIS in the original hypergraph $H$.  Let $H_i$ be the hypergraph being colored in round $i$, where by round we mean one iteration of the \textbf{while} loop or the last call we make either to \algo{\textsc{BL}} or \algo{\textsc{KUW}}. We use the fact that the set $I'$ returned either by \algo{\textsc{BL}} or \algo{\textsc{KUW}} is indeed an MIS in $H_i$, so every violation of independence or maximality in $H_i$ leads to a contradiction.

If the final set of blue vertices is not independent then there is some round $i$ 
of \algo{\textsc{SBL}} in which some edge $e$ became fully blue. This means that a nonempty subset $e'$ of $e$ must be an edge in the hypergraph $H_i$, since $e \setminus e'$ is fully blue and $e'$ is not yet colored. But now round $i$ cannot color $e'$ fully blue because it finds an MIS in $H_i$---a contradiction.

If the final set of blue vertices is not maximal, then it means that some red vertex can be recolored blue without violating independence. Let $v$ be such a vertex and suppose that it was colored red in round $i$. Then in $H_i$, there exists a hyperedge $e'$ such that recoloring $v$ blue will make it fully blue which in turn would lead to some edge $e$ in $H$ being fully blue---a contradiction.

\subsection{Analysis of \algo{\textsc{SBL}} }
 We use the \algo{\textsc{BL}} algorithm
as a subroutine and use the following theorem about its performance:

\begin{theorem}\label{thm:Kelsen}
On a hypergraph with $n$ vertices, $m$ edges and dimension $d \leq \frac{\log^{(2)}n}{4 \log^{(3)}n}$, the \algo{\textsc{BL}} algorithm terminates after 
$O((\log n)^{(d+4)!})$ time with probability at least \newline $1- 1/n^{\Theta(\log n \log^{(2)} n)}$. It uses $\mathrm{poly}(m,n)$
processors and can be implemented on EREW PRAM. 
\end{theorem}

The above result is essentially the same as the corresponding statement in \cite{Kelsen1992} when $d=O(1)$. As mentioned before, 
the proof  follows
from a slight modification of the potential function of \cite{Kelsen1992}; it appears in Section \ref{seclarger}.

In the analysis of the running time below we will focus on the number of rounds of \algo{\textsc{SBL}} algorithm. The time for each 
round of \algo{\textsc{SBL}} is dominated by the time for running \algo{\textsc{BL}} in that round. Specifically, notice that the only other call we make is to $\algo{\textsc{KUW}}(H)$ when the number of vertices in $H$ is less than $1/p^2 = n^{2/{\log^{(3)}n}}$. In the worst case, the runtime of the algorithm is linear in the number of vertices, so we get an additional factor of $O(n^{2/{\log^{(3)}n}})$ in our overall runtime.

We begin by setting values of the parameters used in the algorithm:
\begin{itemize}
\item $p := 1/n^\alpha$,
\item $m := n^\beta$,
\item $\alpha := 1/\log^{(3)}n$,
\item $\beta := \frac{\log^{(2)}n}{8 (\log^{(3)}n)^2}$.
\end{itemize}

We will use the following form of the Chernoff bound (see, e.g., \cite{AlonSpencer}).
\newtheorem{lemma}{Lemma} 
\begin{lemma} \label{lem:Chernoff}
Let $X$ be a random variable taking value $1$ with probability $p \in [0,1]$ and value $0$ with probability $1-p$. Then the sum $X_1 + 
\ldots + X_n$ of i.i.d. copies of $X$ for $a > 0$ satisfies
\begin{align*}
\Pr[(X_1+\ldots+X_n) \leq pn -a] \leq e^{-a^2/2pn}.
\end{align*}
\end{lemma}

There are three kinds of events ($A$, $B$, and $C$) that can happen during the execution of the \algo{\textsc{SBL}} algorithm resulting in 
failure or high running
time. We will show that the union of these events has small probability by upper bounding each event separately and then applying the 
union bound. We use $\Pr[B] \leq \Pr[A] + \Pr[B| \neg{A}]$ and similarly for $\Pr[C]$, resulting in the bound
\begin{align}\label{eqn:union}
\Pr[A \vee B \vee C] \leq 3 \Pr[A] + \Pr[B | \neg{A}] + \Pr[C | \neg{A}].  
\end{align}

\begin{enumerate}
 \item With high probability in each round, the fraction of vertices colored is substantial
 and thus the number of rounds is small.
 \item The probability that a large hyperedge is ever fully marked in a round is small and thus all our applications of \algo{\textsc{BL}} algorithm are valid.
 \item  The probability that a run of \algo{\textsc{BL}} algorithm fails in some round is small.
\end{enumerate}

 We now prove these three claims.
 
(1) Denote the number of marked vertices in round $i$ of \algo{\textsc{SBL}} by $n_i$; thus $n_1 = n$, and in round $i$, the \algo{\textsc{BL}} algorithm is 
invoked on a 
hypergraph with $n_i - n_{i+1}$ vertices (the set of marked vertices). Then, for each $i$, by Lemma~\ref{lem:Chernoff}
we have:
\begin{align*}
\Pr[(n_i - n_{i+1}) \leq pn_i/2] \leq e^{-pn_i/8} \leq e^{-1/(8p)}.
\end{align*}
The inequality above holds because \algo{\textsc{SBL}} algorithm maintains that $n_i \geq 1/p^2$ and so $p n_i \geq 1/p$ for all $i$. 
If the above event holds in each round, then the smallest $r$ satisfying $(1-p/2)^r \leq \frac{1}{p^2n}$ is an upper
bound on the number of rounds. 
Setting
$r := \frac{2\log n}{p} $ gives an upper bound on the number of rounds. Then the probability of the event not
holding in some round becomes
$\frac{2 \log n}{p} \cdot  e^{-1/8p} \leq  \frac{1}{n^{\log n}}$, for sufficiently large $n$.

(2) Conditioning on the number of rounds being upper bounded by $r$, the probability that an edge of size at least $d+1$ is fully marked in some round is at most 
$r m p^{(d+1)}$. If we want this probability to be upper bounded by $1/n$ then we can take
\begin{align*}
d := \frac{\log (rmn)}{\log (1/p)} -1.
\end{align*}
Substituting the values of $r, p, m$ as chosen above we get:
\begin{align*}
d &= \frac{\log 2 + \log^{(2)}n}{\log 1/p} + \frac{\log m}{\log 1/p} + \frac{\log n}{\log 1/p} \\
&= \frac{(\log 2 + \log^{(2)}n) \log^{(3)}n}{\log n} + \frac{\beta (\log^{(3)}n) (\log n)}{\log n} + \frac{\log^{(3)}n \log{n}}{\log n} \\
&\leq 2 \beta \cdot \log^{(3)}n\\
&= \frac{\log^{(2)}n}{4 \log^{(3)}n}, 
\end{align*}
where the inequality holds for all sufficiently large $n$. 

(3) Theorem~\ref{thm:Kelsen} gives that the probability of failure of \algo{\textsc{BL}} algorithm in round $i$ is at most 
$\frac{1}{n_i^{\Theta(\log n_i \log^{(2)}n_i)}}$. Using $n_i \geq 1/p^2$, this probability is at most 
$\frac{1}{n^{\log n}}$ for sufficiently large $n$. Hence, conditioning on the number of rounds being upper bounded by $r$, the probability of any one round failing is upper bounded by $r \cdot \frac{1}{n^{\log n}} \leq \frac{1}{n^{(\log n) /2}}$.

Thus the total probability of failure using \eqref{eqn:union} is at most 
$\frac{3}{n^{\log n}} + \frac{1}{n} + \frac{1}{n^{(\log n)/2}} \leq 2/n$, for sufficiently large $n$. 

Now we account for the time taken by the algorithm. The first round takes time at most $(\log pn)^{d^d}$ (with high prob.), and the subsequent
rounds have the same upper bound. Thus the total time is bounded by $r (\log pn)^{d^d}$ (here we are upper bounding 
$(d+4)!$ somewhat crudely by $d^d$ which holds for all sufficiently large $n$).
For our choice of $d$ above we can upper bound this by

\begin{align*}
r (\log n)^{d^d} \leq r (\log n)^{(\log n)^{1/4}} = \frac{2\log n}{p} (\log n)^{(\log n)^{1/4}} = \\
 2n^{1/\log^{(3)}n} (\log n)^{(\log n)^{1/4}+1} \leq n^{2/\log^{(3)}n}, 
\end{align*}
for sufficiently large $n$. 

This completes the analysis.

\section{Analysis of \algo{\textsc{BL}}} \label{Sec3}
In this section, we present a streamlined analysis of \algo{\textsc{BL}} and show that it can accommodate for a larger $d$ while maintaining the running time of $O((\log n)^{(d+4)!})$. 

Before we describe the improvements in the analysis, we give a brief overview of the algorithm. In the first step, each vertex is marked independently at random with some probability $ p $. After the marking step, for any edge that is fully marked, we unmark all its vertices. We add the remaining marked vertices to the independent set and perform a cleanup operation in which we update the vertex and edge set (by trimming them), remove singleton edges and discard all edges that now contain smaller edges as subsets. We then recurse on this new hypergraph. For a pseudocode of the algorithm we refer the reader to Appendix \ref{pseudo}. 

Like usual, the general strategy in upper bounding the number of rounds necessary for the algorithm to finish, is to define an appropriate quantity and show that progress is being made in each round. Intuitively, we can pick one of several such quantities (the number of vertices, the maximum degree of a vertex, the number of edges etc.) and show that it is reduced by a constant fraction every couple of rounds. The trouble comes from the fact that, in the case of hypergraphs and of the \algo{\textsc{BL}} algorithm in particular, none of these quantities are easy to track. For example, the probability that a vertex gets discarded in one round depends on whether it was marked but never participated in a fully marked edge. When it comes to the degree of a vertex, more evolved measures are needed than in the classical graph case, since now, several vertices can participate together in multiple edges. In this context, we define some essential notation. Let $H = (V,E)$ be a hypergraph with dimension $d$. For $\emptyset \neq x \subseteq V$ and an integer $j$ with $1 \leq j \leq d-|x|$, we define the number of edges of size $|x|+j$ that include $x $ as a subset:
\begin{center}
$N_j(x,H) = \{ y \subseteq V : x \cup y \in E \wedge x \cap y = \emptyset \wedge |y|=j \}$.
\end{center}
We also define the normalized degree of $x$ with respect to dimension $|x|+j$ edges
\begin{center}
$d_j(x,H) = (|N_j(x,H)|)^{1/j}$.
\end{center} 
The maximum normalized degree with respect to dimension $i$ edges then becomes
\begin{center}
$\Delta_i(H) = \max\{ d_{i-|x|}(x,H): x \subseteq V \wedge 0 < |x| < i\}$.
\end{center} 
Finally, the maximum normalized degree is defined as
\begin{center}
$\Delta(H) = \max \{ \Delta_i(H): 2 \leq i \leq d \}$.
\end{center} 

At this point, notice that, as noted in \cite{Kelsen1992}, the main bottleneck in the analysis is the migration of higher dimensional edges to lower dimensional ones. Specifically, in each round, we need to account for the decrease in $N_j(x,H)$ due to edges of size $|x|+j$ decreasing in size, but also for the potential increase due to edges of size $|x|+k$, $k>j$ becoming edges of size $|x|+j$. In order to upper bound such an increase, \cite{Kelsen1992} develops a bound on the upper tail of sums of dependent random variables defined on the edges of a hypergraph. We mention the general bound here and defer the description of its application to later in the paper. 

In order to state the result, we need to describe the probabilistic setting: we consider a hypergraph $H = (V(H),E(H))$ with $n(H)$ vertices, $m(H)$ edges and dimension $dim(H)$. We also consider a weight function $w$ on its edges $w(e)>0$ for any edge $e$. The random variables $C_v$ will correspond to each vertex being colored independently at random with probability $p$ for the color blue and $1-p$ for the color red. Alternatively, the random variable will take the value $1$ with probability $p$ and value $0$ with probability $1-p$. The random variable whose upper tail we will bound will be expressed as the polynomial $S(H,w,p)$. The terms of this polynomial will correspond to an edge $e$ being fully colored blue $C_e = \prod_{v \in e} C_v$. The weights $w(e)$ will become the corresponding coefficients. The polynomial $S(H,w,p)$ then represents the sum of all the weighted edges being colored blue:
\begin{center}
 $S(H,w,p) = \sum\limits_{e \in E(H)} w(e) \cdot C_e$.
\end{center}
Unlike general concentration bounds, we will not compare $S(H,w,p)$ just against its expectation. We will, instead, consider the expected values of all partial derivatives of the polynomial $S(H,w,p)$ with respect to subsets of vertices. Specifically, for a given $x \subseteq V(H)$, we will consider quantities of the form
\begin{center}
$P(H,w,p,x) = \sum\limits_{\substack{e \in E(H)\\ x \subseteq e}} w(e) \cdot p^{|e|-|x|}$.
\end{center}
 Essentially, this term represents the expected sum of the weighted edges around $x$ that are colored all blue, given that $x$ is already colored blue. Notice that this is the same setting used by more recent and considerably better concentration inequalities (e.g. \cite{KimVu},\cite{SchudySviridenko}) to describe their results and in that sense, Kelsen's bound is surprisingly advanced. We then define:
\begin{center}
 $D(H,w,p) = \max \{P(H,w,p,x) : x \subseteq V(H)\}.$
 \end{center}
Notice that $D(H,w,p)$ is greater than the expectation of $S(H,w,p)$. The final result follows:

\begin{theorem} (Theorem $1$ in \cite{Kelsen1992})\label{thm:Kelsenbound}
Let $(H,w)$ be a weighted hypergraph with $dim(H) = d > 0$ and $n(H)=n \geq 3$. For $0 < p \leq 1$ and $\delta > 1$, we have
	\begin{center}
	 $\Pr[S{(H,w,p)} > k(H) \cdot D(H,w,p) ] < p(H)$
	\end{center}
where 
\begin{center}
	$k(H) = ( \log n + 2)^{2^d-1} \cdot \delta^{2^{d-1}}$ and \\
	$p(H) = ( 2^d \cdot \lceil \log n \rceil \cdot m(H))^{d-1} \cdot \log n \cdot (\frac{4e}{\delta-1})^{(\delta-1)/4}$.
\end{center}
\end{theorem}

We are now ready to describe the complete analysis. In order to prove \textbf{Theorem \ref{thm:Kelsen}}, we present a succinct version of the analysis that emphasizes the main ingredients of the proof and our contribution. For full details of the original analysis, we refer the reader to the papers of Beame and Luby \cite{Beame1990} and Kelsen \cite{Kelsen1992}. In the following subsection, we will revisit some of the tools used and show that the analysis goes through even when we consider a higher sampling probability. 

\subsection{Theorem ~\ref{thm:Kelsen}} \label{seclarger}

The main purpose of \textbf{Theorem \ref{thm:Kelsen}} is to show that the analysis follows even when we allow $d \leq \frac{\log^{(2)}n}{4 \log^{(3)}n}$. We start by setting the initial sampling probability to $ p = 1/(a\Delta)$ where $a  = 2^{d+1}$.
The first crucial step is lower bounding the probability that a particular set of vertices $X$ is added to the independent set. We begin by defining random variables $C_v$ for when a vertex $v$ is initially marked (i.e. $C_v=1$ when the $v$ is marked and $0$ otherwise) and $E_v$ for when a vertex is unmarked later due its participation in fully marked edges (i.e. $E_v=1$ when $v$ is unmarked, $0$ otherwise). We also define the random variable $A_v = C_v \wedge \neg{E_v}$  to stand for when the vertex $v$ gets added to the independent set. This notion can be extended to subsets of vertices, by defining $C_X = \bigwedge_{v \in X} C_v$ and $E_X = \bigvee_{v \in X} E_v$, and $A_X = C_X \wedge \neg{ E_X}$. Notice that
\begin{center}
$\Pr[A_X] = \Pr[C_X] \cdot (1-\Pr[E_X | C_X])$.
\end{center}  \emph{Lemma 1} from \cite{Beame1990} shows that $\Pr[A_X] > 1/2 \cdot p^{|X|}$ by proving that $
\Pr[E_X | C_X] < 1/2$. 

\begin{lemma} (Lemma $1$ in \cite{Beame1990})
Given a hypergraph $H=(V,E)$ of dimension $d$, and a set of vertices $X \subseteq V$ with $|X|< d$ such that no $e \subset X$ is an edge, we have 
	$\Pr[E_X | C_X] < 1/2$. (I.e. given that $X$ is marked, it will be added to the IS with probability  $> 1/2$.)
\end{lemma}

We will use the preceding lemma to ensure that progress is being made at each stage of the algorithm. Specifically, we will focus our attention on those sets $X$ that have a large degree with respect to edges of size $|X|+j$. To this extent, \emph{Lemma 2} in \cite{Beame1990}) shows that if such a large degree set exists, then one of the edges that contains it is likely to decrease and turn $X$ into an edge by itself. Once that event occurs, the degree of $X$ becomes $0$.
\begin{lemma} \label{lemma2} (Lemma $2$ in \cite{Beame1990}) For any set of vertices $X$ and $j$ such that $|X| + j \leq d$, if $d_j(X,H) \geq \epsilon \Delta$, then
\begin{center}
 $Pr[\exists Y \in N_j(X,H) : A_Y] \geq \frac{1}{4} (\epsilon/a)^j$,
\end{center}
where $a = 2^{d+1}$.

\end{lemma}

We now discuss the last ingredient of the proof: the upper bound on the migration of edges from higher dimensions to lower dimensions. Notice that the previous lemma is not enough to show that the degree of $X$ will become $0$ in a polylogarithmic number of stages. This is because over each stage, $d_j(X,H)$ can actually increase through the migration of edges from $N_k(X,H)$ where $k>j$. In this context, we employ \textbf{Theorem \ref{thm:Kelsenbound}}. The hypergraph $H'$ we construct consists of all the vertices in $H$ and has as edges all subsets of size $k-j$ of the elements in $N_k(X,H)$, i.e all the potential ways in which an edge of size $|X|+k$ can lose $k-j$ vertices and become an edge of size $|X|+j$ around $X$. Formally, let $X_{j,k}$ be the edge set:
\begin{center}
	$X_{j,k} = \{ Y : Y \subseteq V(H') \wedge |Y| = k-j \wedge \exists Z \in N_k(X,H'), Y\subseteq Z\}$.
\end{center}

The random variables $C_v$ correspond to the situation in which a vertex $v$ gets marked, with probability $p$.  The weight $w'$ of each edge $Y \in X_{j,k}$ represents the number of edges of size $|X| + j$ around $X$ that would be formed if $Y$ were to be fully added to the MIS. Formally:
\begin{center}
 $w'(Y) = |N_j(X \cup Y, H')|$.
\end{center}
 The polynomial $S(H',w',p)$ then becomes an upper bound on the potential increase in $N_j(X,H)$ due to edges in $N_k(X,H)$. Notice that, in our case, $H'$ has dimension at most $d-1 <  \frac{\log^{(2)}n}{4 (\log^{(3)}n)}$, but by choosing $\delta = \log^2 n$, we can arrive at a cleaner formulation of \emph{Theorem 1} from \cite{Kelsen1992}:
\newtheorem{corollary}{corollary} 
\begin{corollary} (Corollary 1 in \cite{Kelsen1992}) Fix a $d > 0$ and a real number $p$, $0<p\leq 1$. For any weighted hypergraph $(H',w)$ of dimension at most $d$ with at most $n$ vertices,
	\begin{center}
	 $\Pr[S{(H',w',p)} > (\log n)^{2^{d+1}} \cdot D(H',w',p) ] < \frac{1}{n^{\Theta(\log n \cdot \log \log n)}}$.
	\end{center}
\end{corollary}
When it comes to $D(H',w',p)$, we can bound it by something more meaningful in our context:
\begin{lemma} (Lemma 3 in \cite{Kelsen1992}) Let $H',w'$ and $p$ be defined as above. Then:
	\begin{center}
		$D(H',w',p) \leq (\Delta_{|X|+k}(H))^j.$
	\end{center}
\end{lemma}

Notice that the same bound applies when we consider the increase in the normalized degree $d_j(X,H')$ and since $\Delta_{|X|+k}(H) \leq \Delta$, we obtain the following \emph{Corollary 3} from \cite{Kelsen1992}:
\begin{corollary} \label{corr3} (Corollary 3 in \cite{Kelsen1992})  With high probability, for $2 \leq j \leq d$, the maximum increase in $d_{j-|X|}(X,H)$ for any non-empty $X \subseteq V$ during a single stage of the algorithm is less than
	\begin{center}
		$\displaystyle{\sum_{k>j}} (\log n)^{2^{k-j+1}} \cdot \Delta_k(H)$.
	\end{center}
\end{corollary} 
Notice that this bound is meaningful in comparison with the trivial bound we would obtain by considering the worst case scenario of \textit{all} higher dimensional edges migrating down:
\begin{center}
 $(\sum_{k>j} \Delta_k(H)^{k-|X|})^{1/(j-|X|)} \geq \sum_{k>j} \Delta_k(H)$,
\end{center}
since $\Delta_k(H)$ could be as high as $n$. 

At this point in the analysis, we can describe the behaviour of each individual $d_j(X,H)$ by a lower bound on the probability that it diminishes when it is too large (\textbf{Lemma \ref{lemma2}}) and an upper bound on how much it can increase in each stage (\textbf{Corollary \ref{corr3}}). We would like to be able to somehow compare these quantities with a universal threshold that we can show will eventually decrease. The trouble comes from expressing the latter of the quantities in terms of this universal threshold: if we compare each $\Delta_k(H)$ to the threshold in the same way (suppose by saying that it is smaller than $1/2$ of the threshold value), we obtain a trivial upper bound on the increase in $\Delta_j(H)$. A solution to this problem would be to define an individual threshold for each $\Delta_k(H)$ separately and relate all of these back to a universal threshold. In this context, 
\cite{Kelsen1992} defines the values $v_i(H)$ inductively by $v_d(H) = \Delta_d(H)$ and:
\begin{center}
	$v_i(H) = \max \{ \Delta_i(H), (\log n)^{f(i)} \cdot v_{i+1}(H) \}$,
\end{center}
for $2 \leq i < d$, where $f$ is a carefully chosen function (to be defined later) that accommodates for the increase in $\Delta_j(H)$ due to migration from higher edges. Essentially, $v_i(H)$ tries to take into account the most significant term in the increased $\Delta_i(H)$: it is either the $\Delta_i(H)$ from the previous round or the most significant term from larger edges offset by a scaling factor $(\log n)^{f(i)} \cdot v_{i+1}(H)$. These individual thresholds relate to the universal threshold by considering the quantities $T_j = v_2(H) / (\log n)^{F(j-1)}$, where $F(i) = \sum_{j=2}^{i} f(j)$ for $2 \leq i \leq d$. Notice that for any hypergraph $H'$, $v_i(H') \leq v_2(H') / (\log n)^{F(j-1)}$. The rest of the analysis focuses on showing that the universal threshold $v_2(H)$ is reduced by a constant fraction every several rounds. 

Let $H_s = (V(H_s),E(H_s))$ be the hypergraph used in stage $s$ of the algorithm and let $v_i = v_i(H_{s_0})$ be the values of these potential functions at the start of a fixed stage of the algorithm. Similarly, let $T_j$ be defined with respect to $v_j$. The main technical lemma is the following:
\begin{lemma} (Lemma 4 in \cite{Kelsen1992}) Let $r$ be an arbitrary positive constant. Then, with high probability, at any stage $s$ with $s_0 \leq s \leq s_0 + (\log n)^r$, we have
	\begin{center}
		$v_2(H_s) \leq v_2 \cdot (1+o(1))$.
	\end{center}
\end{lemma}

In fact, \cite{Kelsen1992} proves that something stronger holds with high probability:
\begin{center}
		$v_j(H_s) \leq T_j \cdot (1+ \lambda(n))$,
\end{center}
where  $\lambda(n) = 2 \cdot \log^{(2)} n / \log n$.

The main argument is by induction on $d-j$ and we will not reproduce it entirely. We will, instead, give the general intuition and focus on the parts of the argument that could change if we allow $d$ to be non-constant. Notice that
$v_j(H_s) =  \max \{ \Delta_j(H_s), (\log n)^{f(j)} \cdot v_{j+1}(H_s) \}$. By induction,
\begin{align*}
 (\log n)^{f(j)} \cdot v_{j+1}(H_s) &\leq  (\log n)^{f(j)} \cdot T_{j+1} \cdot (1+ \lambda(n)) \\
  &\leq T_j \cdot (1+ \lambda(n)).
\end{align*}

So we only need to focus on showing that 
\begin{center}
$ \Delta_j(H_s) \leq T_j \cdot (1+ \lambda(n))$,
\end{center}
with high probability. The tactic is to show that, if $\Delta_j(H_s)$ ever becomes greater than $\frac{1}{2} \cdot T_j \cdot (1+ \lambda(n))$, then in $q_j$ consecutive stages it will decrease with high probability, taking into account the potential migration of edges during those stages. Specifically, suppose there exists an $x \subseteq V(H_s)$ such that
\begin{center}
	$d_{j-|x|}(x,H_s) \geq \frac{1}{2} \cdot T_j \cdot (1+ \lambda(n))$.
\end{center}
One can show that this implies that
\begin{center}
	$d_{j-|x|}(x,H_s) \geq \frac{\Delta(H_s)}{2(\log n)^{F(j-1)}}$.
\end{center}
At this point, we can apply \textbf{Lemma \ref{lemma2}} with $\epsilon = \frac{1}{2(\log n)^{F(j-1)}}$ and get that
\begin{align*}
 \Pr[d_{j-|x|}(x,H_{s+1}) >0] \leq \\
 1 - \frac{1}{2^{d(d+1)}\cdot (\log n)^{F(j-1)(j-1)}}.
\end{align*}
In other words, the probability that in the next round we still have a high normalized degree is small. Notice that if we repeat the argument for
\begin{center}
$q_j = 2^{d (d+1)} \cdot (\log\log n) \cdot (\log n)^{F(j-1)(j-1)+2}$
\end{center}  stages, we have that this remains true with probability at most $1/n^{\Theta(\log n \log\log n)}$. This is the first place in which we differ from the conventional analysis in \cite{Kelsen1992} since we cannot ignore the $2^{d (d+1)}$ factor because it is not constant any more. 

The only step left missing is to guarantee that the increase in $d_{j-|x|}(x,H_{s})$ during those $q_j$ stages is not large. We apply \textbf{Corollary \ref{corr3}} and get that the total increase is $q_j \cdot 	\displaystyle{\sum_{k>j}} (\log n)^{2^{k-j+1}} \cdot \Delta_k(H_s)$. We want to show that such an increase is smaller than $\lambda(n) \cdot T_j$ and since by the inductive assumption we have that $\Delta_k(H_s) \leq T_k \cdot (1+\lambda(n))$, we are left to show that
\begin{center}
	$q_j \cdot \displaystyle{\sum_{k>j}} (\log n)^{2^{k-j+1}} \cdot T_k \cdot (1+\lambda(n)) \leq \lambda(n) \cdot T_j$.
\end{center}

After some calculation, plugging in the values of $q_j$ and $\lambda(n)$, this can be shown to reduce itself to:
\begin{center}
	$2^{d(d+1)} \cdot \displaystyle{\sum_{k>j}} (\log n)^{2^{k-j+1} +F(j-1)\cdot j - F(k-1) +2} \leq \frac{2}{\log n + 2\log \log n}$.
\end{center}

It is at this point that the definition of $f$ comes into play. \cite{Kelsen1992} define $f(2)=7$ and $f(i) = (i-1) \cdot \sum_{j=2}^{i-1} f(i)+7$ for $i>2$. We then get that $F(i) = i \cdot F(i-1)+7$ for $i\geq 2$ and $F(1)=0$. Notice that this definition of $F$ does not allow us to make the above argument. Consider the case when $k=j+1$. Then  
\begin{center}
$2^{k-j+1} +F(j-1)\cdot j - F(k-1) +2 = -1$
\end{center}
and the claim becomes
\begin{center}
 $2^{d(d+1)} \leq \frac{\log n}{\log n + 2\log \log n}$.
\end{center}
This is not true for the larger value of $d$ we are considering. Notice that this was not an issue in the original analysis, because $2^{d(d+1)}$ was a constant in the case they were considering.  


In order for the claim to be true, a different definition of $f$ is required. Specifically, we define the recurrence relationship to be
\begin{center}
 $f(i) = (i-1) \cdot \sum_{j=2}^{i-1} f(i)+d^2$.
\end{center}
In this context, we obtain that $F(i) = i \cdot F(i-1)+d^2$. The claim then becomes:
\begin{center}
	$2^{d(d+1)} \cdot \displaystyle{\sum_{k>j}} (\log n)^{2^{k-j+1}+2-d^2+ F(j) - F(k-1)} \leq \frac{2}{\log n + 2\log \log n}$.
\end{center}
 We will now show that the claim is true for this new definition of $f$. 

We will begin by first noticing that, for any $j$, the highest term in the sum is achieved for $k=j+1$. Formally:
\begin{lemma} For any $k>j+1$ and any $j \geq 2$, we have
\begin{center}
 $2^{k-j+1}+2-d^2 +F(j) - F(k-1)\leq 6 - d^2 $.
 \end{center}
\end{lemma}
The proof is done by showing that the terms are decreasing as a function of $k$ and therefore, the maximum is achieved for the lowest possible value of $k$: $j+1$.

As a consequence, the entire left hand side of the inequality can be upper bounded by
\begin{center}
	$2^{d(d+1)} \cdot (d-j) \cdot \frac{1}{(\log n)^{d^2-6}}$.
\end{center}
By taking $2^{d(d+1)} \leq e^{d(d+1)}$ and $d-j< \log \log n$, it would be enough to show 
\begin{center}
	$e^{d(d+1)} \cdot \frac{1}{(\log n)^{d^2-6}} \leq \frac{1}{\log^2 n} $.
\end{center}
In other words, we can show that
\begin{center}
	$d(d+1) \leq (\log \log n) \cdot (d^2-8)$.
\end{center}
One can check that this inequality holds for $d <\frac{\log^{(2)}n}{4 \log^{(3)}n}$.

At this point, we have shown that the total increase in $d_{j-|x|}(x,H_{s})$ during those stages is upper bounded with high probability by $\lambda(n) \cdot T_j$. Moreover, notice that after $q_j$ stages,  $\Delta_j(H)$ will not exceed $ T_j \cdot \frac{1+3\lambda(n)}{2}$ and hence, after $q_d$ stages, we have that, with high probability, 
\begin{center}
	$v_j(H_{s_1}) \leq  T_j \cdot \frac{1+3\lambda(n)}{2}$
\end{center}
for any $2 \leq j \leq d$ and $s_0 \leq s_1 \leq s_0 + q_j$. In fact, going back to the start of the algorithm, $v_2(H)$ is reduced by a constant factor, with high probability, every $q_d$ stages. Hence, after $O(\log n \cdot q_d)$ stages, we have that $v_2(H) = 0$ and therefore, $V(H)=0$ and the algorithm terminates. 

Now we are left to prove that 
\begin{center}
 $\log n \cdot q_d \leq (\log n)^{(d+4)!}$.
\end{center}
Notice that
\begin{align*}
q_d &\leq (log^{(2)} n)^2 \cdot (\log n)^{F(d-1)(d-1)+2}\\
    &\leq (\log n)^{F(d-1)(d-1)+3} \\
    &\leq (\log n)^{(d+4)!-1}
\end{align*}
where the last inequality can be verified by inductively proving that $F(i) \leq d^2 \cdot (i+2)!$ for all $i$.

\section{Stronger Concentration Bound}
The next step of the proof that we are going to improve is the bound that Kelsen gives on the maximum potential increase in edges in one round, using the same setting as in the original analysis but employing the Kim-Vu concentration bound \cite{KimVu}. We obtain an analogue of \emph{Corollary 2} in \cite{Kelsen1992}:

\begin{corollary} For $X \subseteq V$, and $1 \leq j < k\leq d-|X|$, we have 
	\begin{center}
		$\Pr[S(X,j,k) > (1+a_{k-j} \lambda^{k-j}) \cdot (\Delta_{|X|+k}(H))^j ] \leq 2e^2e^{-\lambda}n^{k-j-1}$,
	\end{center}
	where $a_{k-j}= 8^{k-j} (k-j)!^{1/2}$.
\end{corollary}
Notice that we upper bounded the term $D(H',w',p)$ by $(\Delta_{|X|+k}(H))^j$, just like \cite{Kelsen1992}. Simple algebra can show that this result follows even for the new value of $a$. Choosing $\lambda = \Theta(\log^2 n)$,  we get that an analogue of Corollary $3$ in \cite{Kelsen1992}:
\begin{corollary} With high probability, for $2 \leq j \leq d$, the maximum increase in $d_{j-|X|}(X,H)$ for any non-empty $X \subseteq V$ during a single stage of the algorithm is less than:
	\begin{center}
		$\displaystyle{\sum_{k>j}} (\log n)^{2(k-j)} \cdot \Delta_k(H)$.
	\end{center}
\end{corollary} 
Notice that the bound of $(\log n)^{2(k-j)}$ is much smaller than the one of $(\log n)^{2^{k-j+1}}$ in \cite{Kelsen1992}.

\subsection{Discussion}
In context of this improvement, the natural next step is to investigate what effect it has on the overall running time of \algo{\textsc{BL}}. We show that, under the current set up of the potential function, no improvement is possible. Specifically, we show that the function $F$ must be roughly exponential for the argument to follow, despite the obvious improvements.

Notice that in our previous attempt to call \algo{\textsc{BL}} on a hypergraph with super-constant dimension, the main issue was showing that the increase in $q_j$ rounds was upper bounded by $T_j \cdot \lambda(n)$. Incorporating all of the new improvements, we get that the claim formally looks like:
\begin{center}
	$(5d)^d \cdot \displaystyle{\sum_{k>j}} (\log n)^{2(k-j) +F(j-1)\cdot j - F(k-1) +2} \leq \frac{2}{\log n + 2\log \log n}$.
\end{center}

We proved this claim by showing that the largest term in the sum was upper bounded by $\frac{1}{(\log n)^{d^2-6}}$. We check the minimal conditions that $f$ must satisfy in order for the new claim to be true by precisely looking at this largest term in the case when $k=j+1$ for a fixed $2 \leq j \leq d$. Notice that, first of all, this will be the largest term when we allow $F$ to satisfy $F(k-1) > F(j) + 2(k-j-2)$ for all $k>j$. Given that, the term will be:
\begin{center}
 $(\log n)^{4 + F(j-1) \cdot j - F(j)}$.
\end{center}
Notice that, in order for the claim to be true, this term needs to be smaller than 
\begin{center}
 $ \frac{1}{(5d)^d }\cdot \frac{2}{\log n + 2\log \log n}$.
\end{center}
In order for this to happen, we must have that
\begin{center}
$(\log n)^{4 + F(j-1) \cdot j - F(j)} \leq \frac{1}{\log n}$.
\end{center}
This, in turn, requires that 
\begin{center}
$F(j) \geq F(j-1) \cdot j + 5$.
\end{center}

\section{Conclusion}
In this paper, we build on the RNC algorithm for computing an MIS in constant dimension hypergraphs to get an $n^{o(1)}$ algorithm on general hypergraphs when the number of edges is upper bounded by $ n^{\frac{\log^{(2)}n}{8(\log^{(3)}n)^2}}$. In order to perform the analysis, we prove that the subroutine algorithm can be adapted to run on a larger dimension while maintaining an appropriate running time. We also present independent improvements to the analysis of the latter and identify the main bottleneck in the approach that affects the final runtime most significantly. For example, notice that the factor of $j$ in the above inequality $F(j) \geq F(j-1) \cdot j + 5$ originated from \textbf{Lemma ~\ref{lemma2}}. Specifically, \cite{Beame1990} lower bound the probability that $d_j(X) > \epsilon \Delta$ becomes $0$ in the next iteration, as a function of $(\epsilon/a)^j$. A refinement of that result could potentially lead to a weaker restriction on $F$ and hence, a smaller running time.

\paragraph{Acknowledgments} Authors David G. Harris and Aravind Srinivasan were supported in part by \textbf{NSF Award CNS-1010789}. 

\bibliographystyle{plain}
\bibliography{HMIS}

\begin{thebibliography}{1}

\bibitem{AlonSpencer}
Noga Alon and Joel Spencer.
\newblock {\em The Probabilistic Method}.
\newblock John Wiley, 1992.

\bibitem{Beame1990}
P.~Beame and M.~Luby.
\newblock {Parallel search for maximal independence given minimal dependence}.
\newblock In {\em Proceedings of the first annual ACM-SIAM symposium on
  Discrete algorithms}, pages 212--218. Society for Industrial and Applied
  Mathematics, 1990.

\bibitem{KarpR90}
Richard~M. Karp and Vijaya Ramachandran.
\newblock Parallel algorithms for shared-memory machines.
\newblock In {\em Handbook of Theoretical Computer Science, Volume A:
  Algorithms and Complexity (A)}, pages 869--942. 1990.

\bibitem{Karp1988}
R.M. Karp, E.~Upfal, and A.~Wigderson.
\newblock {The complexity of parallel search}.
\newblock {\em J. Comput. Syst. Sci.}, 36(2):225--253, 1988.

\bibitem{Kelsen1992}
Pierre Kelsen.
\newblock {On the Parallel Complexity of Computing a Maximal Independent Set in
  a Hypergraph}.
\newblock {\em Fourth annual ACM symposium on Theory of computing}, 3:339--350,
  1992.

\bibitem{KimVu}
Jeong~Han Kim and Van~H. Vu.
\newblock Concentration of multivariate polynomials and its applications.
\newblock {\em Combinatorica}, 20(3):417--434, 2000.

\bibitem{LuczakS97}
Tomasz Luczak and Edyta Szymanska.
\newblock A parallel randomized algorithm for finding a maximal independent set
  in a linear hypergraph.
\newblock {\em J. Algorithms}, 25(2):311--320, 1997.

\bibitem{SchudySviridenko}
Warren Schudy and Maxim Sviridenko.
\newblock Concentration and moment inequalities for polynomials of independent
  random variables.
\newblock In {\em SODA}, pages 437--446, 2012.

\bibitem{ShachnaiS04}
Hadas Shachnai and Aravind Srinivasan.
\newblock Finding large independent sets in graphs and hypergraphs.
\newblock {\em SIAM J. Discrete Math.}, 18(3):488--500, 2004.

\end{thebibliography}

\appendix
\section{\algo{\textsc{BL}} algorithm} \label{pseudo}
We give the pseudocode of the \algo{\textsc{BL}} algorithm as initially described in \cite{Beame1990}.
\begin{algorithm}
\caption{\algo{\textsc{BL}}}

\textbf{Input:} A hypergraph $H = (V,E)$ \\
\textbf{Output:} A maximal independent set $I \subseteq V$.

~\

\begin{algorithmic}[1] 
\STATE Calculate $\Delta(H)$ as defined in Section \ref{Sec3}.
\STATE Let $d = \max \{ |e| : e \in E\}$ and $p = 1/ (2^{d+1} \Delta)$.
\STATE $H'=(V',E') \leftarrow H=(V,E)$.
\STATE $I \leftarrow \emptyset$.
\WHILE{$V' \neq \emptyset$}
		\STATE Select vertices independently at random with probability $p$. 
		\STATE Let $I'$ be the collection of such selected vertices.
		\FORALL {$e \in E'$ such that $e \subseteq I'$}
			\STATE $I' \leftarrow I'\setminus e$.
		\ENDFOR
		\STATE $I \leftarrow I \cup I'$.
		\STATE $V' \leftarrow V' \setminus I'$.
		\FORALL{$e \in E'$}
			\STATE $e \leftarrow e \setminus I'$.
		\ENDFOR
		\FORALL{$e,e' \in E'$}
			\IF{$e \subseteq e'$}
				\STATE $E' \leftarrow E' \setminus e$.
			\ENDIF
		\ENDFOR	
		\FORALL{$e = \{v\} \in E'$}
			\STATE $E' \leftarrow E' \setminus e$.
			\STATE $V' \leftarrow V' \setminus \{v\}$.
		\ENDFOR
\ENDWHILE
\STATE Return $I$.

\end{algorithmic}
\end{algorithm}

\end{document}